\documentclass[aps,showpacs,amssymb,amsfonts,superscriptaddress,twocolumn,prl]{revtex4-1}
\usepackage{bm,bbm}
\usepackage{times}
\usepackage{graphicx,amsthm,epstopdf,amsmath}
\usepackage{subfigure,color}
\usepackage[draft]{fixme}
\usepackage{tikz}

\newcommand{\beq}[0]{\begin{equation}}
\newcommand{\eeq}[0]{\end{equation}}
\newcommand{\bw}[0]{\begin{widetext}}
\newcommand{\ew}[0]{\end{widetext}}
\newcommand{\bc}[0]{\begin{center}}
\newcommand{\ec}[0]{\end{center}}
\newcommand{\bwn}[0]{\begin{widetext}\begin{eqnarray}}
\newcommand{\ewn}[0]{\end{eqnarray}\end{widetext}}
\newcommand{\beqn}[0]{\begin{eqnarray}}
\newcommand{\eeqn}[0]{\end{eqnarray}}

\newcommand{\proj}[1]{|#1\rangle\!\langle #1|}
\newcommand{\ket}[1]{|#1\rangle}

\def\calB{{\cal B}}

\def\calT{{\cal T}}

\definecolor{myurlcolor}{rgb}{0,0,0.7}
\definecolor{myrefcolor}{rgb}{0.8,0,0}
\usepackage{hyperref}
\hypersetup{colorlinks, linkcolor=myrefcolor,
citecolor=myurlcolor, urlcolor=myurlcolor}

\newtheorem{thm}{Theorem}
\newtheorem*{main*}{Main result}
\newtheorem*{thm*}{Theorem}

\newtheorem{lem}[thm]{Lemma}

\newcommand{\ot}[0]{\otimes}

\newcommand{\beu}{\begin{equation}}
\newcommand{\eeu}{\end{equation}}

\newcommand{\be}{\begin{eqnarray}}
\newcommand{\ee}{\end{eqnarray}}

\newcommand{\ba}{\begin{array}}
\newcommand{\ea}{\end{array}}

\newcommand{\Tr}[0]{\mathrm{Tr}}

\begin{document}

\title{Entanglement and nonlocality are inequivalent for any number of
particles}

\author{R. Augusiak}
\affiliation{ICFO--Institut de Ciencies Fotoniques, 08860
Castelldefels (Barcelona), Spain}

\author{M. Demianowicz}
\affiliation{ICFO--Institut de Ciencies Fotoniques, 08860
Castelldefels (Barcelona), Spain}

\author{J. Tura}
\affiliation{ICFO--Institut de Ciencies Fotoniques, 08860
Castelldefels (Barcelona), Spain}

\author{A. Ac\'in}
\affiliation{ICFO--Institut de Ciencies Fotoniques, 08860
Castelldefels (Barcelona), Spain}
\affiliation{ICREA--Instituci\'o
Catalana de Recerca i Estudis Avan\c{c}ats, Lluis Companys 23,
08010 Barcelona, Spain}

\begin{abstract}
Understanding the relation between nonlocality and entanglement is
one of the fundamental problems in quantum physics.
In the bipartite case, it is known that the correlations observed
for some entangled quantum states can be explained within the
framework of local models, thus proving that these resources are
inequivalent in this scenario. However, except for a
single example of an entangled three-qubit state that has a local
model, almost nothing is known about such relation in multipartite systems. We
provide a general construction of genuinely multipartite entangled
states that do not display genuinely multipartite nonlocality,
thus proving that entanglement and nonlocality are inequivalent
for any number of particles.
\end{abstract}
\maketitle

\textit{Introduction.} Nonlocality -- the phenomenon of
the impossibility of describing  by local hidden variable (LHV)
theories the correlations arising from measuring quantum states --
is a fundamental characteristics  of quantum mechanics. It was
Bell who first pointed out that there exist quantum states for
which underlying  classical variables cannot account for the
measurement statistics on them \cite{Bell}. Such states are
called nonlocal and they violate Bell
inequalities (see \cite{przegladowka-dani}).
%
%
Having been confronted with the result of Bell one
might be tempted to identify nonlocality with entanglement, another
essential resource of quantum information theory.  This intuition
is, however, not correct as the relationship between these two
notions is more involved: while LHV models trivially exist for
separable states, not all entangled states are nonlocal.

The first step in the exploration of this inequivalence
was taken by Werner in Ref. \cite{Werner}.
There, he introduced a family of highly symmetric states, nowadays
known as the Werner states, and provided an explicit LHV model
reproducing the measurement statistics obtained when two parties
perform local projective measurements on
some states from this family.
Building on this model, Barrett proved that there are entangled
Werner states that remain local even if general measurements are
performed \cite{Barrett}. Both these models
were later adapted to mixtures of any state $\rho$ and
the white noise \cite{Almeida}, such as the isotropic states
for which $\rho$ is the maximally entangled state
\cite{isotropic}. The nonlocal properties of noisy states have
also been related to an important mathematical notion known as the
Grothendieck constant in \cite{Groth}.

\begin{figure}[h!]
\includegraphics[width=0.5\columnwidth]{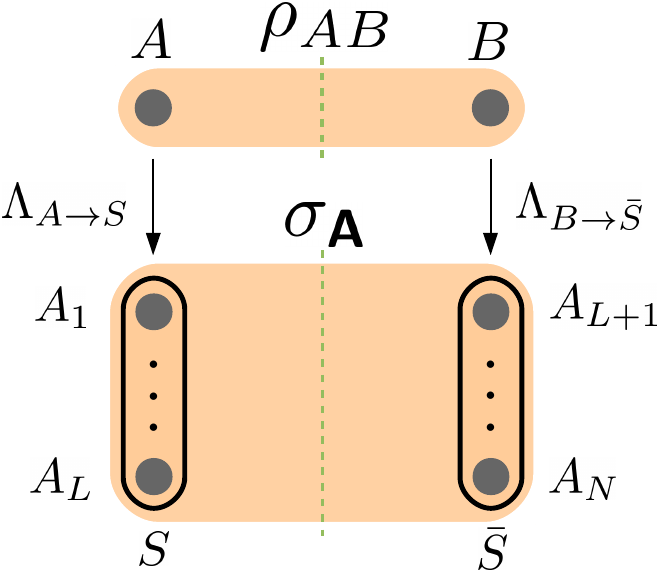}
\caption{\textbf{Schematic depiction of our construction.} The
initial state $\rho_{AB}$ is an entangled bipartite state with a
local model for generalized measurements. The application of
properly chosen local channels $\Lambda_{A\to S}$ and
$\Lambda_{B\to\bar{S}}$ with $S=A_1\ldots A_L$ and
$\bar{S}=A_{L+1}\ldots A_N$ to both subsystems of $\rho_{AB}$
creates an $N$-partite state $\sigma_{\boldsymbol{\mathsf{A}}}$
that is genuinely multipartite entangled and bilocal with respect
to the bipartition $S|\bar{S}$. The dashed green line determines
the partitions with respect to which both $\rho_{AB}$ and
$\sigma_{\boldsymbol{\mathsf{A}}}$ are local.}\label{fig:scheme}
\end{figure}

Very little is known about the relation between entanglement and nonlocality in
the general multipartite scenario. Here, this question becomes much
subtler as the multiparty scenario offers a richer
variety of different types of entanglement and nonlocality. For
instance, for any number of parties, it is trivial to construct a
non-separable, and thus entangled state, that has a local model:
it suffices to combine an entangled and local Werner state for two
parties with a fully product state for the remaining ones.
However, this is clearly nothing but a manifestation of the
inequivalence between entanglement and locality for two parties.
Thus, the most natural question in the multipartite scenario is
whether for an arbitrary number of parties $N$, there always
exists a gap between genuinely $N$-party entanglement and
genuinely $N$-party nonlocality. Our main goal
is to show that this is the case, thus proving that entanglement
and nonlocality are inequivalent for {\it any} number of parties.

The departure point of our proof are bipartite entangled states
with local models for general measurements. Basing on them, we
provide a general construction of genuinely multipartite entangled
(GME) states of $N$ parties that do not display genuinely
multipartite nonlocality. In fact, our construction is such that
the resulting $N$-party entangled state has a bilocal model, in
which the parties are divided into two groups, inherited from the
local model of the initial bipartite state (see Fig.
\ref{fig:scheme}). We also generalize our construction to map any
$N$-party GME state with a $K$-party local model to an $N'$-party
GME state, with $N'>N$, that also has a $K$-local model.

Before moving to the proof of our results, it is worth mentioning
that for three parties T\'oth and Ac\'in found a fully local model
for arbitrary projective measurements on a genuinely entangled
three--qubit state  \cite{TothAcin}. This model has been later
extended to general measurements
\cite{Tulio}, proving ultimately that there are GME states of three parties that do not display any form of nonlocality.
However, it is unknown whether this model can be extended to more parties. 

\textit{Preliminaries.} Let us start by introducing some notation
and terminology. Consider $N$ parties
$\boldsymbol{\mathsf{A}}:=A_1,\ldots,A_N$ (for small $N$ also
denoted $A$, $B$, etc.) sharing an $N$-partite quantum state
$\rho_{\boldsymbol{\mathsf{A}}}\in \mathcal{B}(\mathcal{H}_{N,d})$,
where $\mathcal{H}_{N,d}=(\mathbbm{C}^d)^{\ot N}$ and
$\mathcal{B}(\mathcal{H})$ stands for the set of bounded linear
operators acting on the Hilbert space $\mathcal{H}$. Moreover, by
$\mathcal{S}_{M,d}$ and $P^{X}_{\mathrm{sym}}$ we
denote, respectively, the symmetric subspace of
$\mathcal{H}_{M,d}$ and the projector onto the symmetric
subspace of the subsystem $X$.

Let us then partition the parties
into $K$ pairwise disjoint groups $S_i$ such that by adding them one
recovers $\boldsymbol{\mathsf{A}}$,
and call it a $K$-partition; for $K=2$ we call it a bipartition
and denote $S|\bar{S}$ with $\bar{S}$
being the complement of $S$ in $\boldsymbol{\mathsf{A}}$. Denoting
by $\mathcal{S}_K$ the set of all $K$-partitions, we say that
$\rho_{\boldsymbol{\mathsf{A}}}$ is {\it $K$-separable} ({\it
biseparable} for $K=2$) if it is a probabilistic mixture of
$N$-partite states separable with respect to some
$K$-partitions, i.e.,
\begin{equation}\label{separable}
 \rho_{\boldsymbol{\mathsf{A}}}=\sum_{S\in\mathcal{S}_K}p_S\sum_{i}
q_S^i\bigotimes_{k=1}^K\proj{\psi_{S_k}^i}.
\end{equation}
Here, $p_S$ and $q_S^i$ are probability distributions and
$\ket{\psi_{S_k}^i}$ are pure states defined on the $S_k$ subsystem.
One then calls $\rho_{\boldsymbol{\mathsf{A}}}$ \textit{fully
separable} if it is $N$-separable, and, \textit{genuinely
multipartite entangled} (GME) if it does not admit any of the
above forms of separability; in particular, it is not biseparable.

Analogously to the notions of separability one introduces those of
locality. Imagine that on their share of the state
$\rho_{\boldsymbol{\mathsf{A}}}$, each party $A_i$ performs a
measurement $M_i=\{M_{a_i}^{(i)}\}$, where $a_i$ enumerate the
outcomes of $M_i$, while $M_{a_i}^{(i)}$ are the measurement
operators, i.e., $M^{(i)}_{a_i}\ge 0$ and
$\sum_{a_i}M^{(i)}_{a_i}=\mathbbm{1}_d$. If, additionally,
$M_{a_i}^{(i)}$ are supported on orthogonal subspaces, we call the
corresponding measurement \textit{projective} (PM), and
\textit{generalized} (GM; also called POVM) otherwise. Now,
adopting the definition from Ref. \cite{Svetlichny}, one says that
the state $\rho_{\boldsymbol{\mathsf{A}}}$ is \textit{$K$-local
for GMs}, or shortly \textit{$K$-local}, if for any choice of
measurements $\boldsymbol{\mathsf{M}}:=M_1,\ldots,M_N$, the
probability of obtaining the outcomes
$\boldsymbol{\mathsf{a}}:=a_1,\ldots, a_N$ decomposes as
\begin{equation}\label{Kmodel}
p(\boldsymbol{\mathsf{a}}|\boldsymbol{\mathsf{M}})=\sum_{S\in\mathcal{S}_K}p_{S}
\int\mathrm{d}\lambda\,\omega_{S}(\lambda)\prod_{k=1}^Kp_k(\boldsymbol{
\mathsf{a}}_{S_k}
|\boldsymbol{\mathsf{M}}_{S_k},\lambda).
\end{equation}
Here, the sum goes over all possible $K$-partitions of
$\boldsymbol{\mathsf{A}}$, $p_{S}$ and $\omega_{S}$ are
probability distributions, and $p_k(\cdot|\cdot)$ is the
probability (also called \textit{response function}) that the
parties belonging to $S_k$ obtain $\boldsymbol{\mathsf{a}}_{S_k}$
upon measuring $\boldsymbol{\mathsf{M}}_{S_k}$, while having the
\textit{classical information} $\lambda$.
Accordingly, a state $\rho_{\boldsymbol{\mathsf{A}}}$ admitting
(\ref{Kmodel}) is said to have a \textit{$K$-local model}. In
particular, if $K=N$ we say that the state is \textit{fully
local}, while if $K=2$ --- \textit{bilocal}. Notice that there are
also local models reproducing projective measurements only, and
mixed ones, i.e., those reproducing GMs for some parties and PMs
for the rest. By comparing (\ref{separable}) and (\ref{Kmodel}),
it is direct to realize that every $K$--separable state is
$K$--local. Multipartite states which do not admit any form
of bilocality are called genuinely multipartite nonlocal (GMN).

It should be noted that such definition of $K$-locality,
Eq. (\ref{Kmodel}), has been shown to be inconsistent
with an operational interpretation of nonlocality given in Refs
\cite{gallego,bancal}. However, we use it here because of its
direct analogy to entanglement. Moreover, it allows us to state
our construction in a general way and facilitates the proof of our
result. Nevertheless, as we argue below, the inequivalence between
entanglement and nonlocality also holds for these operational
definitions.

\textit{Inequivalence of entanglement and nonlocality.}
%
We are now in position to state and
prove our main result.
\begin{main*}\label{glowne}
Entanglement and nonlocality are inequivalent
for any number of parties $N$, as for any $N$ there exist genuinely
entangled $N$--partite states with bilocal models.
\end{main*}

To prove the result we proceed in two steps. First, we show that
any bipartite local state can be converted into a multipartite
state with a bilocal model. Then, we argue that such construction
may lead to GME states for any $N$.

As to the first step, we generalize the
observation made by Barrett \cite{Barrett}. Let $\varrho_{AB}\in
\mathcal{B}(\mathcal{H}_{2,d})$ be arbitrary and let
\begin{equation}
\Lambda_{A\to S}:\mathcal{B}(\mathbbm{C}^{d})\to
\mathcal{B}((\mathbbm{C}^{d'})^{\ot L})
\end{equation}
and
\begin{equation}
\Lambda_{B\to\bar{S}}:\mathcal{B}(\mathbbm{C}^{d})\to
\mathcal{B}((\mathbbm{C}^{d'})^{\ot N-L})
\end{equation}
be a pair of quantum channels sending operators acting
on a single-party Hilbert space $\mathbbm{C}^d$ to operators
acting on $L$-partite and $(N-L)$-partite Hilbert spaces of local
dimension $d'$, respectively, with $S=A_1\ldots A_{L}$ and
$\bar{S}=A_{L+1}\ldots A_{N}$. Now, one can prove the following lemma.
\begin{lem}\label{Lem1}
 If $\rho_{AB}$  has a local model for generalized
measurements, then, for any pair of quantum channels
$\Lambda_{A\to S}$ and  $\Lambda_{B\to \bar{S}}$ defined above,
the $N$-partite state
\begin{equation}\label{stan-sigma}
\sigma_{\boldsymbol{\mathsf{A}}}=(\Lambda_{A\to
S}\ot\Lambda_{B\to \bar{S}})(\rho_{AB})
\end{equation}
has a bilocal model for any measurements.
\end{lem}
\begin{proof} The reasoning is analogous to the one by Barrett from Ref.
\cite{Barrett}, but for completeness we present it here.

The fact that $\rho_{AB}$ has a local model for generalized
measurements means that the probabilities of obtaining results
$a,b$ upon performing measurements $M_A=\{M^A_a\}$ and
$M_B=\{M^B_b\}$, respectively, by the parties $A$ and $B$, assume
the ``local'' form (\ref{Kmodel}), which for $N=2$ simplifies to
\begin{eqnarray}\label{model}
 \hspace{-0.4cm}p(a,b|M_A,M_B)&=&
 \!\int\!\mathrm{d}
\lambda\,\omega(\lambda)p_{\rho}(a|M_A,\lambda)p_{\rho}(b|M_B,\lambda).
\end{eqnarray}
Here we have used the subscript $\rho$ to emphasize that the
probabilities correspond to $\rho_{AB}$. Exploiting this model,
we will now demonstrate that $\sigma_{\boldsymbol{\mathsf{A}}}$ is
bilocal with respect to the bipartition $S|\bar{S}=A_1\ldots
A_L|A_{L+1}\ldots A_N$. To this end, let us assume that the
parties perform measurements $M_i=\{M_{a_i}^{(i)}\}$,
$i=1,\ldots,N$, on their shares of the state
$\sigma_{\boldsymbol{\mathsf{A}}}$. Then, denoting by
$\Lambda^{\dagger}_{S\to A}$ and $\Lambda^{\dagger}_{\bar{S}\to
B}$ the dual maps of $\Lambda_{A\to S}$ and $\Lambda_{B\to
\bar{S}}$ \cite{footnote}, respectively, we define the following operators
\begin{equation}\label{opus1}
 \bar{M}^{A}_{\boldsymbol{\mathsf{a}}_S}=\Lambda_{S\to
A}^{\dagger}\left(\bigotimes_{i=1}^{L}M^{(i)}_{a_i}\right), \quad
\bar{M}^B_{\boldsymbol{\mathsf{a}}_{\bar{S}}}=\Lambda_{\bar{S}\to
B}^{\dagger}\left(\bigotimes_{i=L+1}^{N}M^{(i)}_{a_{i}}\right)
\end{equation}
acting on $\mathbbm{C}^{d}$ and indexed by the outcomes
$\boldsymbol{\mathsf{a}}_S:=a_1,\ldots,a_L$ and
$\boldsymbol{\mathsf{a}}_{\bar{S}}:=a_{L+1},\ldots,a_N$. Since the
dual map of a quantum channel is positive and unital (it preserves
the identity operator), it is direct to see that the operators
(\ref{opus1}) form generalized measurements,
denoted $\bar{M}_A$ and $\bar{M}_B$. With their aid, let us now
define the response functions for the state
$\sigma_{\boldsymbol{\mathsf{A}}}$ corresponding to the parties
$A_1,\ldots,A_L$ and $A_{L+1},\ldots,A_N$, respectively, as
$p_{\sigma}(\boldsymbol{\mathsf{a}}_S|\boldsymbol{\mathsf{M}}_S,\lambda)=p_{\rho
}(\boldsymbol{\mathsf{a}}_S|\bar{M}_A,\lambda)$
and
$p_{\sigma}(\boldsymbol{\mathsf{a}}_{\bar{S}}|\boldsymbol{\mathsf{M}}_{\bar{S}},
\lambda)=p_{\rho}(\boldsymbol{\mathsf{a}}_{\bar{S}}|\bar{M}_B,\lambda).$
Then,
\begin{eqnarray}\label{calculation}
p(\boldsymbol{\mathsf{a}}|\boldsymbol{\mathsf{M}})
&&=\Tr[(M^{(1)}_{a_1}\ot\ldots\ot
M^{(N)}_{a_N})\sigma_{\boldsymbol{\mathsf{A}}}]\nonumber\\
&&=\Tr[(M^{(1)}_{a_1}\ot\ldots\ot
M^{(N)}_{a_N})(\Lambda_{A\to S}\ot\Lambda_{B\to\bar{S}})(\rho_{AB})]\nonumber\\
&&=\Tr[\bar{M}^A_{a_1\ldots a_{L}}\ot \bar{M}^B_{a_{L+1}\ldots
a_N}\rho_{AB}]\nonumber\\
&&=\int\mathrm{d}
\lambda\,\omega(\lambda)p_{\rho}(\boldsymbol{\mathsf{a}}_S|\bar{M}_A,\lambda)p_{
\rho }(\boldsymbol{\mathsf{a}}_{\bar{S}}|\bar{M}_B,\lambda)\nonumber\\
&&=\int\mathrm{d}\lambda\,\omega(\lambda)p_{\sigma}(\boldsymbol{\mathsf{a}}
_S|\boldsymbol{\mathsf{M}}_S,\lambda)p_{\sigma}(\boldsymbol{\mathsf{a}}_{\bar{S}
}|\boldsymbol{\mathsf{M}}_{ \bar{S}},
\lambda),
\end{eqnarray}
where we have utilized Eqs. (\ref{opus1}) and
the definition of $\sigma_{\boldsymbol{\mathsf{A}}}$.
It thus follows that
$\sigma_{\boldsymbol{\mathsf{A}}}$ has a
bilocal model for GMs with respect to $A_1\ldots A_L|A_{L+1}\ldots
A_N$.
\end{proof}

The critical point of our approach will
be to observe that the above mapping of local bipartite
states to bilocal multipartite ones may lead to GME states. To argue
this, we need a technical result concerning genuine multipartite entanglement.

\begin{lem}\label{Lem2}
Consider an $N$-partite state
$\sigma_{\boldsymbol{\mathsf{A}}}\in\mathcal{B}(\mathcal{H}_{N,d})$ and
assume that with respect to some bipartition $S|\bar{S}$, the
subsystems $S$ and $\bar{S}$ are symmetric, that is,
$P_{\mathrm{sym}}^S\otimes
P_{\mathrm{sym}}^{\bar{S}}\sigma_{\boldsymbol{\mathsf{A}}}
P_{\mathrm{sym}}^S\otimes
P_{\mathrm{sym}}^{\bar{S}}=\sigma_{\boldsymbol{\mathsf{A}}}$
holds. If $\sigma_{\boldsymbol{\mathsf{A}}}$ is not GME, then it
is biseparable with respect to this bipartition, i.e.,
\begin{equation}\label{Lem3:1}
\sigma_{\boldsymbol{\mathsf{A}}}=\sum_{i}p_i \sigma^i_S\ot \sigma^i_{\bar{S}}
\end{equation}
with $\sigma^i_S$ and $\sigma^i_{\bar{S}}$ being states
defined on subsystems $S$ and $\bar{S}$.
\end{lem}
\begin{proof}As the proof is rather technical and lengthy, here we present
its sketch moving the details to Appendix.

The assumption that $\sigma_{\boldsymbol{\mathsf{A}}}$ is not
GME means that it admits the decomposition
(\ref{separable}) with $K=2$, i.e.,
%
$\sigma_{\bold{A}}=\sum_{T|\bar{T}\in\mathcal{S}_2}p_{T|\bar{T}}
\rho_{T|\bar{T}}$.
%
The sum goes over all bipartitions $T|\bar{T}$ of $\boldsymbol{\mathsf{A}}$ and
$\rho_{T|\bar{T}}$ is some state separable with respect to $T|\bar{T}$, i.e.,
%
$\rho_{T|\bar{T}}=\sum_{i}q_{T|\bar{T}}^i \proj{e_T^i}\ot
 \proj{f^i_{\bar{T}}}.$
%
%
Now, exploiting the assumption that the subspaces $S$ and
$\bar{S}$ of $\sigma_{\boldsymbol{\mathsf{A}}}$ are symmetric, one
can prove that each $\rho_{T|\bar{T}}$ with $T\neq S$ in the
decomposition of $\sigma_{\boldsymbol{\mathsf{A}}}$ is of the form
(\ref{Lem3:1}) ($\sigma_{S|\bar{S}}$ is already of this form). To
this aim, it is enough to realize that every pure state
$\ket{e_T^i}\ket{f_{\bar{T}}^i}$ must obey
$P_{\mathrm{sym}}^{S(\bar{S})}\ket{e_T^i}\ket{f_{\bar{T}}^i}=\ket{e_T^i}\ket
{f_{ \bar{T}}^i} $. This, after some algebra, implies that it must
also be product with respect to $S|\bar{S}$.
\end{proof}

We are now in position to prove our main result. A straightforward
corollary of Lemma \ref{Lem2} is that any $N$-partite state
$\sigma_{\boldsymbol{\mathsf{A}}}$, which does not admit the form
(\ref{Lem3:1}), i.e., is entangled across some cut $S|\bar{S}$,
and whose subsystems $S$ and $\bar{S}$ are symmetric, is GME. Take
now a bipartite entangled state $\rho_{AB}\in
\mathcal{B}(\mathbbm{C}^d\otimes \mathbbm{C}^d)$ and the quantum
channels $\Lambda_{A\to S}:\mathcal{B}(\mathbbm{C}^d)\to
\mathcal{B}(\mathcal{S}_{L,d'})$ and
$\Lambda_{B\to\bar{S}}:\mathcal{B}(\mathbbm{C}^d)\to
\mathcal{B}(\mathcal{S}_{N-L,d'})$ that are invertible in the
sense that for both of them there exists a channel
$\widetilde{\Lambda}$ such that $\widetilde{\Lambda}\circ\Lambda$
is the identity map on $\mathcal{B}(\mathbbm{C}^d)$. Note that now
these channels output states acting on the corresponding $L$ and
$(N-L)$-partite symmetric subspaces. Clearly, the $N$-partite
state $\sigma_{\boldsymbol{\mathsf{A}}}$ resulting from the
application of $\Lambda_{A\to S}$ and $\Lambda_{B\to \bar{S}}$ to
$\rho_{AB}$ is symmetric on the subspaces $S$ and $\bar{S}$, and,
as $\rho_{AB}$ is entangled, must be GME; if
$\sigma_{\boldsymbol{\mathsf{A}}}$ is not GME, then, as the two
channels are invertible, $\rho_{AB}$ must be separable. If we
further assume that $\rho_{AB}$ is local, the resulting state
$\sigma_{\boldsymbol{\mathsf{A}}}$ will have, according to Lemma
\ref{Lem1}, a bilocal model, proving the desired.

As a result we have a general method for constructing bilocal genuinely
entangled $N$-partite states with an arbitrary $N$.

\textit{Applications.} Let us now see how
our method works in practice. We
consider for this purpose two paradigmatic classes of states: the isotropic and
the Werner states \cite{isotropic,Werner}. The quantum channels are chosen to be
$\Lambda_{A\to S}(\cdot)=V_L(\cdot)V^{\dagger}_L$ and
$\Lambda_{B\to \bar{S}}(\cdot)=V_{N-L}(\cdot)V^{\dagger}_{N-L}$
with $V_M:\mathbbm{C}^d\to \mathcal{S}_{M,d}$ being an isometry
defined through $V_M\ket{i}=\ket{i}^{\ot M}$ for any element of
the standard basis in $\mathbbm{C}^d$.

Let us begin with the two-qudit isotropic states
which are given by
$\rho_{\mathrm{iso}}(p)=p\proj{\psi_{d}^+}+(1-p)\mathbbm{1}_{d^2}/d^2$, where
$\ket{\psi_d^+}=(1/\sqrt{d})\sum_{i=0}^{d-1}\ket{ii}$ is the
maximally entangled state. Application of the isometries to
$\rho_{\mathrm{iso}}(p)$ leads us to the mixture of the well-known
GHZ state of $N$ qudits
$\ket{\mathrm{GHZ}_{N,d}}=(1/\sqrt{d})\sum_{i=0}^{d-1}
\ket{i}^{\otimes N}$ and some coloured noise:
\begin{equation}
\label{genghz}
 \sigma_{\boldsymbol{\mathsf{A}}}(p)=p\proj{\mathrm{GHZ}_{N,d}}+(1-p)\frac{
\mathcal{P}_{L,d}\ot\mathcal{P}_{N-L,d}}{d^2},
\end{equation}
where $\mathcal{P}_{L,d}=\sum_{i=0}^{d-1}\proj{i}^{\ot L}$ with
$1\leq L\leq N-1$. Now, as the isotropic states are local for
$p\leq (3d-1)(d-1)^{d-1}/d^d(d+1)$ \cite{Almeida}, it stems from
Lemma \ref{Lem1} that for the same range of $p$ and
$L=1,\ldots,N-1$, the state $\sigma_{\boldsymbol{\mathsf{A}}}(p)$
is bilocal with respect to the bipartition $A_1\ldots
A_L|A_{L+1}\ldots A_N$. Further, isometric
channels are always invertible ($V_M^{\dagger}V_M=\mathbbm{1}_d$)
and thus, as required, preserve entanglement. Hence, the states
$\sigma_{\boldsymbol{\mathsf{A}}}(p)$ are GME for the same range of $p$ as
$\rho_{\mathrm{iso}}(p)$ are entangled, i.e., for
$p>1/(d+1)$. Concluding, the states (\ref{genghz}) constitute our
first example of GME states with a bilocal model for any $N$.

Let us now consider the Werner states which read
$\rho_{W} (p)=p[2/d(d-1)]P_{\mathrm{asym}} +(1-p)
\mathbbm{1}_{d^2}/d^2$,
where $P_{\mathrm{asym}}$ stands for the projector onto the
antisymmetric subspace of $\mathbbm{C}^d\ot\mathbbm{C}^d$.
Applying the isometries defined above to $\rho_{W}(p)$,
one constructs the following $N$-qudit states
\begin{equation}
 \sigma'_{\boldsymbol{\mathsf{A}}}(p)=p\frac{2\widetilde{P}_{L,d}}{d(d-1)}+(1-p)
\frac{\mathcal{P}_{L,d}\ot \mathcal{P}_{N-L,d}}{d^2},
\end{equation}
with $L=1,\ldots,N-1$, where
%
$\widetilde{P}_{L,d}=\sum_{i<j}\proj{\psi_{ij}}$
%
with $\ket{\psi_{ij}}=(1/\sqrt{2})(\ket{i}^{\ot L}\ket{j}^{\ot
(N-L)}+\ket{j}^{\ot L}\ket{i}^{\ot (N-L)})$. The Werner states
have a local model for GMs for $p\leq (3d-1)(d-1)^{d-1}/d^d(d+1)$
\cite{Barrett} and so do the states
$\sigma'_{\boldsymbol{\mathsf{A}}}(p)$ for any $L$. Moreover,
$\rho_{W}(p)$ are entangled for $p>1/(d+1)$, thus
$\sigma'_{\boldsymbol{\mathsf{A}}}(p)$ are GME for the same range
of $p$.

\textit{Generalizing the construction.} Interestingly, our
construction can be generalized to the case when
the initial bipartite state is replaced by a
multipartite genuinely entangled state with a local model for GMs.
To be precise, let us first consider a $K$-partite state
$\rho_{A_1\ldots A_K}$ acting on $\mathcal{H}_{K,d}$ and a
collection of $K$ quantum channels $\Lambda_{A_k\to
S_k}:\mathcal{B}(\mathbbm{C}^d)\to \mathcal{B}((\mathbbm{C}^{d'})^{\ot L_k})$
with $L_i\geq 1$ such that $L_1+\ldots+L_K=N>K$. By definition,
each channel ``expands'' a
single-particle Hilbert space to an $L_k$-partite one of local
dimension $d'$ corresponding to parties from the group
%
 $S_k=\{A_{1+\sum_{i=1}^kL_{i-1}},\ldots,A_{\sum_{i=1}^{k}L_i}\},$
%
with $k=1,\ldots,K$ and $L_0=0$.
By applying these channels to subsystems of $\rho_{A_1\ldots A_K}$, one obtains
an $N$-partite state
\begin{equation}
 \widetilde{\sigma}_{\boldsymbol{\mathsf{A}}}=(\Lambda_{A_1\to
S_1}\ot\ldots\ot\Lambda_{A_K\to S_K})(\rho_{ A_1\ldots A_K})
\end{equation}
 acting on $\mathcal{H}_{N,d'}$ (the subsystem
$A_1$ of $\rho_{A_1\ldots A_K}$ is mapped to $S_1=A_1\ldots
A_{L_1}$ of $\widetilde{\sigma}_{\boldsymbol{\mathsf{A}}}$ etc.).
Now, following the same arguments as in the proof of Lemma
\ref{Lem1} (see Appendix), one shows that if $\rho_{A_1\ldots
A_K}$ has a fully local model for GMs, then
$\widetilde{\sigma}_{\bold{A}}$ has a $K$-local model for GMs with
respect to the $K$-partition determined by the groups $S_k$.
Furthermore, generalizing Lemma \ref{Lem2} (see Appendix), one
finds that with a proper choice of the channels $\Lambda_{A_k\to
S_k}$ we can guarantee that the resulting state is GME. Thus, any
genuinely entangled $K$-partite state admitting a fully local
model gives rise to a genuinely entangled $N$-partite state, with
any $N>K$, having $K$-local model.

Note that this generalization, when applied to the existing
example of a tripartite GME state with a local model for
generalized measurements \cite{Tulio}, implies the existence of
GME $N$-partite states with three-local models for any $N$.

Finally, let us comment on the operational definitions of
$K$-locality given in Refs.
\cite{gallego,bancal}. As shown there, a definition of
$K$-locality which is operationally consistent looks like
(\ref{Kmodel}) with the additional constraint that all the
probability distributions
$p_k(\boldsymbol{\mathsf{a}}_{S_k}
|\boldsymbol{\mathsf{M}}_{S_k},\lambda)$ satisfy the no-signalling
principle. Importantly, this condition can be easily met in our
construction. Namely, it is enough to extend the part of a state
on which the response function in a local model is quantum, i.e.,
given by the Born rule, as the resulting response on this part
will be automatically no-signalling. The isotropic
and the Werner states do have
local models with one of the response functions being quantum,
so the states $\sigma_{\boldsymbol{\mathsf{A}}}$ and
$\sigma_{\boldsymbol{\mathsf{A}}}'$ constructed
above with $L=1$ have bilocal models in which the
response functions corresponding to $N-1$ parties are
no-signalling. Thus, the inequivalence between entanglement and
nonlocality holds even when using the operational definitions of
$K$-locality.

\textit{Conclusions and discussion.}
We have provided a general method of deriving from $N$-party GME states with a
$K$-local models $N'$-party GME states with the same type of locality for any
$N'>N$. Our construction implies then that entanglement and nonlocality are
inequivalent for any number of parties, even if the operational definitions of
multipartite locality are considered.

The most interesting open problem following from our work is to understand
the extent to which the inequivalence between entanglement and nonlocality
holds. With the current state of knowledge, our
results show that there exist GME $N$-party states that have a 2-
and even 3-local model. Now, what is the maximum value of $K$ such
that there exist $N$-party states with a $K$-local model for any
$N$? In particular, are there genuinely entangled $N$-party states with a fully
local model? This happens to be the case for $N=2,3$, but no results are
known beyond these two cases. A related question is whether there
exists some threshold value of $N$ above which the GME states are too entangled
to allow for a fully local model.

We then note that a number of different operationally meaningful
nonlocality scenarios beyond the one considered in the present work have been
introduced. These are: the network approach \cite{Cavalcanti}, Bell scenarios defined on
copies of a state \cite{Palazuelos} (see also Ref. \cite{Liang}) or sequential
measurements
\cite{Popescu,Hidden-Gisin,Hirsch}. In these more general approaches, states that
are local in the standard setup may display nonlocal properties. Nevertheless,
it remains open whether in such scenarios the equivalence between
nonlocality and entanglement holds. It would be thus of interest to verify
whether nonlocality of the states introduced here could be revealed in one of
these more general setups.

Let us conclude by pointing out that our construction also implies
that genuine multipartite entanglement is inequivalent to steering---another
intriguing phenomenon of quantum information theory \cite{steering}. That is, by
applying it to a bipartite state that has a local model with
quantum response function, the construction produces a GME state which is
unsteerable (in at least one direction) across the same bipartition with respect
to which it is bilocal.

\textit{Acknowledgments.} This work is supported by the EU ERC CoG
QITBOX and AdG OSYRIS, EU project SIQS, Spanish project Chist-Era DIQIP, and the John
Templeton Foundation. R. A. also acknowledges the Spanish MINECO
for the support through the Juan de la Cierva program.

\section{appendices}

Here we state and prove generalizations of Lemma 1 and 2 from the
main text. We also formulate our main result as a theorem.

We start with a generalization of Lemma 1. For this purpose,
consider a $K$-partite state $\rho_{A_1\ldots A_K}$ acting on
$\mathcal{H}_{K,d}=(\mathbbm{C}^d)^{\ot K}$ and a collection of
$K$ quantum channels
\begin{equation}\label{StanBorys}
 \Lambda_{A_k\to S_k}:\mathcal{B}(\mathbbm{C}^d)\to
\mathcal{B}((\mathbbm{C}^{d'})^{\ot L_i})\qquad (i=1,\ldots,K),
\end{equation}
such that $L_i\geq 1$ for any $i$ and $L_1+\ldots+L_K=N>K$. Each
channel $\Lambda_{A_k\to S_k}$ maps operators acting on the
Hilbert space $\mathbbm{C}^d$ corresponding to the party $A_k$ to
operators acting on $L_i$-partite Hilbert space
$(\mathbbm{C}^{d'})^{\ot L_i}$ corresponding to the group of
parties denoted
\begin{equation}\label{partition}
 S_k=\{A_{1+\sum_{i=1}^kL_{i-1}},\ldots,A_{\sum_{i=1}^{k}L_i}\}.
\end{equation}
with $k=1,\ldots,K$ and $L_0=0$. Notice that in the latter Hilbert
space the local dimension $d'$ may be different than $d$.
Consequently, application of these channels to the subsystems of
$\rho_{A_1\ldots A_K}$ gives rise to an $N$-partite state
\begin{eqnarray}\label{StanSigma}
\sigma_{\boldsymbol{\mathsf{A}}}=(\Lambda_{ A_1\to
S_1}\ot\ldots\ot \Lambda_{A_K\to S_K})(\rho_{A_1\ldots A_K})
\end{eqnarray}
that now acts on a larger $N$-partite Hilbert space
$\mathcal{H}_{N,d'}$. This mapping induces naturally the
$K$-partition of $\boldsymbol{\mathsf{A}}$ determined by the
groups $S_k$ given in (\ref{partition}).

We are now ready to generalize Lemma 1 from the main text.
\begin{lem}\label{LemA1}
 If a $K$-partite state $\rho_{A_1\ldots A_K}$ acting on $\mathcal{H}_{K,d}$ has
a fully local model for GMs, then the $N$-partite state
$\sigma_{\boldsymbol{\mathsf{A}}}$ in Eq. (\ref{StanSigma}) with
$L_i\geq 1$ and $L_1+\ldots +L_K=N>K$ has a $K$-local model for
GMs with respect to the $K$-partition defined by
(\ref{partition}).
\end{lem}
%
%
%
\begin{proof}
We start by noting that the fact that $\rho_{A_1\ldots A_K}$ has a
fully local model for generalized measurements means that the
probabilities of obtaining results $a_1,\ldots,a_K$ upon
performing measurements $M_i=\{M_{a_i}^{(i)}\}$ $(i=1,\ldots,K)$
take the following form [cf. Eq. (2) in the main text]
\begin{eqnarray}\label{model}
&& \hspace{-0.5cm}p(a_1,\ldots,a_K|M_1,\ldots,M_K)\nonumber\\
&& =\Tr\left[\left(M^{(1)}_{a_1}
\ot\ldots \ot  M_{a_K}^{(K)}\right)\rho_{A_1\ldots A_K}\right]\nonumber\\
&& =\int_{\Omega}\mathrm{d}
\lambda\,\omega(\lambda)P_{\rho}(a_1|M_{a_1}^{(1)},\lambda)\ldots
P_{\rho}(a_K|M_{a_K}^{(K)},\lambda),\nonumber\\
\end{eqnarray}
where $\Omega$ denotes the set over which the classical
information $\lambda$ is distributed with probability distribution
$\omega$, and we have used the subscript $\rho$ to emphasize that
the probabilities correspond to $\rho_{A_1\ldots A_K}$. Exploiting
(\ref{model}), we can now construct a local model for
$\sigma_{\boldsymbol{\mathsf{A}}}$ with respect to the
$K$-partition defined by the groups $S_k$ given in Eq.
(\ref{partition}). To this end, let us assume that the parties
$A_1,\ldots,A_N$ perform measurements
$\widetilde{M}_{i}=\{\widetilde{M}_{a_i}^{(i)}\}$ $(i=1,\ldots,N)$
on their share of the state $\sigma_{\boldsymbol{\mathsf{A}}}$.
Then, we define the following operators
\begin{eqnarray}\label{ops1}
\bar{M}_{\boldsymbol{\mathsf{a}}_{S_1}}^{(1)}&\!\!=\!\!&\Lambda_{S_1\to
A_1}^{ \dagger } \left(\widetilde{ M}^{(1)}_{a_1}\ot\ldots \ot
\widetilde{M}^{(L_1)}_{a_{L_1}}\right)\nonumber\\
\bar{M}^{(2)}_{\boldsymbol{\mathsf{a}}_{S_2}}&\!\!=\!\!&\Lambda_{S_2\to
A_2}^{ \dagger } \left(\widetilde{M}^{ (L_1+1)}_{a_{L_1+1}} \ot
\ldots \ot
\widetilde{M}^{(L_1+L_2)}_{a_{L_1+L_2}}\right)\nonumber\\
&\vdots&\nonumber\\
\bar{M}^{(K)}_{\boldsymbol{\mathsf{a}}_{S_K}}&\!\!=\!\!&\Lambda_{S_K\to
A_K}^{ \dagger } \left(\widetilde{M}^{
(L_1+\ldots+L_{K-1}+1)}_{a_{L_1+\ldots+L_{K-1}+1}}\ot \ldots \ot
\widetilde{M}^{(N)}_{a_N}\right)\nonumber\\
\end{eqnarray}
acting on $\mathbbm{C}^{d}$ and indexed by the collections of
outcomes $\boldsymbol{\mathsf{a}}_{S_k}$. Here,
$\Lambda^{\dagger}_{S_k\to
A_k}:\mathcal{B}((\mathbbm{C}^{d'})^{\ot L_k})\to
\mathcal{B}(\mathbbm{C}^d)$ stands for the dual map of
$\Lambda_{A_k\to S_k}$. Due to the fact that the dual map of a
quantum channel is unital, i.e., it preserves the identity, and
positive, it is fairly easy to see that for each $k=1,\ldots,K$,
the set of operators
$\bar{M}_k=\{\bar{M}^{(k)}_{\boldsymbol{\mathsf{a}}_{S_k}}\}_{
\boldsymbol{\mathsf{a}}_{S_k}} $ forms a generalized measurement,
that is, for any $k=1,\ldots,K$,
$\bar{M}^{(k)}_{\boldsymbol{\mathsf{a}}_{S_k}}\geq 0$ for any
$\boldsymbol{\mathsf{a}}_{S_k}$ and
\begin{equation}
 \sum_{\boldsymbol{\mathsf{a}}_{S_k}}\bar{M}^{(k)}_{\boldsymbol{\mathsf{a}}_
{S_k}}=\mathbbm{1}_d.
\end{equation}
Using the measurements $\bar{M}_k$, let us now define the response
functions corresponding to the sets of parties $S_k$ as
\begin{equation}\label{expr1}
p_{\sigma}(\boldsymbol{\mathsf{a}}_{S_k}|\boldsymbol{\mathsf{M}}_{
S_k},\lambda)=p_{\rho}
(\boldsymbol{\mathsf{a}}_{S_k}|\bar{M}_k,\lambda)
\end{equation}
with $k=1,\ldots,K$. We now have
\begin{eqnarray*}
&&p(\boldsymbol{\mathsf{a}}|\boldsymbol{\widetilde{\mathsf{M}}})=\Tr\left[
\left(\widetilde{M}^{(1)}_{a_1}\ot\ldots\ot
\widetilde{M}^{(N)}_{a_N}\right)\sigma_{\boldsymbol{\mathsf{A}}}\right]
\nonumber\\
&&=\Tr\left[\left(\widetilde{M}^{(1)}_{a_1}\ot\ldots\ot
\widetilde{M}^{(N)}_{a_N} \right)
\bigotimes_{k=1}^K\Lambda_{A_k\to S_k}(\rho_{A_1\ldots
A_K})\right]
\end{eqnarray*}
where we have employed the definition of
$\sigma_{\boldsymbol{\mathsf{A}}}$. This, with the aid of Eqs.
(\ref{ops1}) can be further rewritten as
\begin{equation}
p(\boldsymbol{\mathsf{a}}|\boldsymbol{\widetilde{\mathsf{M}}})=\Tr\left[
\left(\bar{M}_{\boldsymbol{\mathsf{a}}_{S_1}}^{(1)}\ot\ldots\ot
\bar{M}_{\boldsymbol{\mathsf{a}}_{S_K}}^{(K)}\right)\rho_{A_1\ldots
A_K}\right].
\end{equation}
Since the state $\rho_{A_1\ldots A_K}$ is fully local, we finally
obtain
\begin{eqnarray}
p(\boldsymbol{\mathsf{a}}|\boldsymbol{\widetilde{\mathsf{M}}})&=&\int_{\Omega}
\mathrm{d}
\lambda\,\omega(\lambda)\prod_{k=1}^{K}P_{\rho}(\boldsymbol{\mathsf{a}}_{S_k}
|\bar{M}_{k},\lambda)\nonumber\\
&=&\int_{\Omega}\mathrm{d}
\lambda\,\omega(\lambda)\prod_{k=1}^{K}P_{\sigma}(\boldsymbol{\mathsf{a}}_{S_k}
|\boldsymbol{\mathsf{M}}_{S_k},\lambda),
\end{eqnarray}
where the last equality stems from the definitions of the response
functions for the state $\sigma_{\boldsymbol{\mathsf{A}}}$ given
in Eq. (\ref{expr1}). As a result, one sees that
$\sigma_{\boldsymbol{\mathsf{A}}}$ has a $K$-local model for
generalized measurement with respect to the $K$-partition defined
by the groups $S_k$ given in (\ref{partition}). Notice that
$\Omega$ and $\omega$ are the same as in the local model for
$\rho_{A_1\ldots A_K}$.
%
%
%
%
\end{proof}
%
Let us remark that it is fairly easy to see that Lemma \ref{LemA1}
can be further generalized to the case when the initial Hilbert
spaces $\mathcal{H}_{K,d}$ (and analogously the final one
$\mathcal{H}_{N,d'}$) have different local dimensions.

We can now move to the discussion on genuine multipartite
entanglement. For this purpose let us consider again $N$ parties
$A_1,\ldots,A_N$ sharing some $N$-partite state
$\rho_{\boldsymbol{\mathsf{A}}}$ acting on
$\mathcal{H}_{N,d}=(\mathbbm{C}^d)^{\ot N}$. Let us then consider
some $K$-partition of the parties $A_1,\ldots,A_N$, into $K$
pairwise disjoint sets $S_k$ such that they together contain all
the parties. Finally, by $P_\mathrm{sym}^{X}$ we denote the
projector onto the symmetric subspace of the Hilbert space
corresponding to the subsystem $X$.

One then proves the following generalization of Lemma 2 from the
main text.

\begin{lem}\label{LemA2}
Let $\rho_{\boldsymbol{\mathsf{A}}}$ be an $N$-partite state
acting on $\mathcal{H}_{N,d}$ such that with respect to some
$K$-partition given by the groups $S_k$, its subsystems
corresponding to $S_k$ are defined on symmetric subspaces, i.e.,
\begin{equation}
 P_{\mathrm{sym}}^{S_k}\rho_{\boldsymbol{\mathsf{A}}}
P_{\mathrm{sym}}^{S_k}=\rho_{\boldsymbol{\mathsf{A}}}
\end{equation}
with $k=1,\ldots,K$. If $\rho$ is not GME, then it takes the
biseparable form
\begin{equation}\label{biseparable}
 \rho_{\boldsymbol{\mathsf{A}}}=\sum_{\mathcal{T}|\bar{\mathcal{T}}}p_{\mathcal{
T}|\bar{\mathcal{T}}}\rho_{ \mathcal{T|\bar{\mathcal{T}}}}
%
\end{equation}
where $p_{\mathcal{T}|\bar{\mathcal{T}}}$ is some probability
distribution and every $\rho_{\mathcal{T}|\bar{\mathcal{T}}}$  is
a state separable across the bipartition
$\mathcal{T}|\bar{\mathcal{T}}$ with $\mathcal{T}$ and
$\bar{\mathcal{T}}$ being unions of $S_k$.
%
\end{lem}

\begin{proof}From the fact that $\rho_{\boldsymbol{\mathsf{A}}}$ is not GME it
follows that it can be written as [cf. Eq. (1) in the main text]
\begin{equation}\label{sepform}
\rho_{\boldsymbol{\mathsf{A}}}=\sum_{T|\bar{T}\in\mathcal{S}_2}p'_{T|\bar{T}}
\varrho_{T|\bar{T}},\qquad p'_{T|\bar{T}}\geq 0, \qquad
\sum_{T|\bar{T}}p'_{T|\bar{T}}=1,
\end{equation}
where the sum goes over all bipartitions $T|\bar{T}$ and
$\varrho_{T|\bar{T}}$ is some state that is separable with respect
to the bipartition $T|\bar{T}$, i.e., it admits the form
\begin{equation}\label{Lem3:4}
 \varrho_{T|\bar{T}}=\sum_{i}q_{T|\bar{T}}^i \proj{e_T^i}\ot
\proj{f^i_{\bar{T}}},
\end{equation}
with $q_{T|\bar{T}}^i$ being some probability distribution for any
$T|\bar{T}$, and $\ket{e_T^i}$ and $\ket{f_T^i}$ denoting some
pure states from the Hilbert spaces corresponding to subsystems
$T$ and $\bar{T}$.

We will now prove, using the assumption that the state
$\rho_{\boldsymbol{\mathsf{A}}}$ is symmetric on the subsystems
$S_k$, that any $T$ and $\bar{T}$ appearing in Eq. (\ref{Lem3:4})
must be a union of the sets $S_k$. To this end, it is enough to
show for each $T$ that every pure state
$\ket{e_T^i}\ket{f_{\bar{T}}^i}$ appearing in (\ref{Lem3:4}) must
also be product across a bipartition
$\mathcal{T}|\bar{\mathcal{T}}$ in which $\mathcal{T}$ and
$\bar{\mathcal{T}}$ are unions of the sets $S_k$.

We first notice that the assumption that the subsystems $S_k$ of
$\rho_{\boldsymbol{\mathsf{A}}}$ are defined on the corresponding
symmetric subspaces implies that any pure state
$\ket{e_T^i}\ket{f_{\bar{T}}^i}$ appearing in (\ref{Lem3:4}) for
any bipartition $T$ must obey the following set of conditions
\begin{equation}
P_{\mathrm{sym}}^{S_k}\ket{e_T^i}\ket{f_{\bar{T}}^i}=\ket{e_T^i}\ket{f_{\bar{T}
} ^i }
\end{equation}
with $k=1,\ldots,K$. This in particular means that for any pair of
parties $A_m$ and $A_n$ belonging to the same set $S_k$,
\begin{equation}\label{condition}
 V_{A_mA_n}\ket{e_T^i}\ket{f_{\bar{T}}^i}=\ket{e_T^i}\ket{f_{\bar{T}}^i},
\end{equation}
where $V$ is the swap operator defined through the condition
$V\ket{\phi}\ket{\psi}=\ket{\psi}\ket{\phi}$ for any pair of
vectors $\ket{\psi},\ket{\phi}\in\mathbbm{C}^d$.

Let us now consider a particular bipartition $T|\bar{T}$ in Eq.
(\ref{sepform}) for which $T$ and $\bar{T}$ are not unions of the
sets $S_k$. Then, there exists a pair of parties $A_m,A_n$
belonging to one of the sets $S_k$ (the same one) such that
$A_m\in T$ and $A_n\in\bar{T}$. For such a pair we use the Schmidt
decompositions of the vectors $\ket{e_T^i}$ and
$\ket{f_{\bar{T}}^i}$,
\begin{equation}\label{Schmidt1}
 \ket{e_T}=\sum_{j}\sqrt{\mu_j}\ket{e_{A_m}^j}\ket{e_{T\setminus A_m}^j},
\end{equation}
and
\begin{equation}\label{Schmidt2}
 \ket{f_{\bar{T}}}=\sum_{j}\sqrt{\nu_j}\ket{f_{A_n}^j}\ket{f_{\bar{T}\setminus
A_n}^j},
\end{equation}
where for simplicity we have skipped the upper index $i$ and the
subscript $T\setminus A_m$ means the subsystem $T$ but the
single-party subsystem $A_m$. Then, the condition
(\ref{condition}) implies
\begin{eqnarray}
&&\sum_{j,j'}\sqrt{\mu_j\nu_{j'}}\ket{f_{A_n}^{j'}}\ket{e_{T\setminus
A_m}^{j}}\ket{e_{A_m}^{j}}\ket{f_{\bar{T} \setminus
A_n}^{j'}}\nonumber\\
&&=\sum_{j,j'}\sqrt{\mu_j\nu_{j'}}\ket{e_{A_m}^{j}}
\ket{e_{T\setminus A_m}^{j}}\ket{f_{A_n}^{j'}}\ket{f_{\bar{T}
\setminus A_n}^{j'}},
\end{eqnarray}
which by virtue of the orthogonality of the vectors in
(\ref{Schmidt1}) and (\ref{Schmidt2}), implies that
$\ket{e_{A_m}^j}=\ket{f_{A_n}^{j'}}$ for any pair of indices
$j,j'$. Denoting then
$\ket{g}=\ket{e_{A_m}^j}=\ket{f_{A_n}^{j'}}$, one finally finds
that
$\ket{e_T^i}\ket{f_{\bar{T}}^i}=\ket{g_{A_m}}\ket{e'_{T\setminus
A_m }}\ket{g_{A_n}}\ket{f'_{\bar{T}\setminus A_n}}$, i.e., every
vector in the decomposition (\ref{Lem3:4}) must be product with
respect to the parties $A_m$ and $A_n$. By repeating this
procedure for all pairs of parties $A_m,A_n$ such that both belong
to one of the sets $S_k$, but $A_m\in T$ and $A_n\in \bar{T}$
(actually, not all pairs are necessary), one finds that every
$\ket{e_T^i}\ket{f_{\bar{T}}^i}$ in the decomposition of
$\varrho_T$ is product with respect to some bipartition
$\mathcal{T}|\bar{\mathcal{T}}$ with $\mathcal{T}$ and
$\bar{\mathcal{T}}$ being unions of groups $S_k$.

By applying exactly the same argument to the remaining
bipartitions, one shows that any $\varrho_{T|\bar{T}}$ in
(\ref{sepform}) is separable with respect to some bipartition
$\mathcal{T}|\bar{\mathcal{T}}$ with $\mathcal{T}$ and
$\bar{\calT}$ being unions of $S_k$, giving us the form
(\ref{biseparable}) and completing the proof.
\end{proof}

Clearly, Lemma 2 from the main text follows directly from the
above one. Precisely, assuming that with respect to some
bipartition $S|\bar{S}$, the state
$\rho_{\boldsymbol{\mathsf{A}}}$ obeys
\begin{equation}
P_{\mathrm{sym}}^S\rho_{\boldsymbol{\mathsf{A}}}
P_{\mathrm{sym}}^{S}=\rho_{\boldsymbol{\mathsf{A}}}
\end{equation}
and
\begin{equation}
 P_{\mathrm{sym}}^{\bar{S}}\rho_{\boldsymbol{\mathsf{A}}}
P_{\mathrm{sym}}^{\bar{S}}=\rho_{\boldsymbol{\mathsf{A}}},
\end{equation}
it follows that if it is not GME, then
$\rho_{\boldsymbol{\mathsf{A}}}$ must be separable across the
bipartition $S|\bar{S}$, that is,
\begin{equation}
 \rho_{\boldsymbol{\mathsf{A}}}=\sum_{i}p_i\rho^i_S\ot
\overline{\rho}^i_{\bar{S}}.
\end{equation}
with $\rho_S^i$ and $\overline{\rho}_{\bar{S}}^i$ being some
states corresponding to the groups $S$ and $\bar{S}$.

The following corollary straightforwardly stems from Lemma
\ref{LemA2}. For any $N$-partite state
$\rho_{\boldsymbol{\mathsf{A}}}$ such that with respect to some
$K$-partition its subsystems corresponding to the sets $S_k$
$(k=1,\ldots,K)$ are symmetric, if $\rho$ does not admit the
decomposition (\ref{sepform}) with $\rho_{T|\bar{T}}$ being
separable across bipartitions $T|\bar{T}$ for which $T$ and
$\bar{T}$ are unions of the sets $S_k$, then $\rho$ is GME. This
fact gives rise to the following theorem, being a generalization
of the main result of our work.

\begin{thm*}
Let $\rho_{A_1\ldots A_K}$ be an entangled state acting on
$\mathcal{H}_{K,d}$ that has a fully local model for generalized
measurements. Then, for any collection of $K$ invertible quantum
channels
\begin{equation}
\Lambda_{A_k\to S_k}:\calB(\mathbbm{C}^d)\to
\calB(\mathcal{S}_{L_k,d'}),
\end{equation}
with $\mathcal{S}_{L_k,d'}$ denoting the symmetric subspace of
$(\mathbbm{C}^{d'})^{\ot L_k}$, the state
\begin{equation}
\sigma_{\boldsymbol{\mathsf{A}}}=(\Lambda_{A_1\to S_1}\ot\ldots\ot
\Lambda_{A_K\to S_K})(\rho_{A_1\ldots A_K})
\end{equation}
has a $K$-local model with respect to the $K$-partition with the
groups $S_k$ defined in Eq. (\ref{partition}), and it is GME.
\end{thm*}
\begin{proof}
From Lemma \ref{LemA1} it follows that the state
$\sigma_{\boldsymbol{\mathsf{A}}}$ has a $K$-local model for
generalized measurements with respect to the given $K$-partition.
Then, the fact that $\rho_{A_1\ldots A_K}$ is genuinely
multipartite entangled implies that so is
$\sigma_{\boldsymbol{\mathsf{A}}}$. To make it more explicit, let
us assume, in contrary, that $\sigma_{\boldsymbol{\mathsf{A}}}$ is
not GME. Due to Lemma \ref{LemA2}, this means that it admits the
decomposition (\ref{biseparable}), with all $T$ and $\bar{T}$
being unions of the sets $S_k$ (by the definition of the channels
$\Lambda_{A_k\to S_k}$, the $S_k$ subsystems of
$\rho_{\boldsymbol{\mathsf{A}}}$ are symmetric). Since all
channels $\Lambda_{A_k\to S_k}$ $(k=1,\ldots,K)$ are invertible,
this implies that $\rho_{A_1\ldots A_K}$ is not GME, contradicting
the assumption.
\end{proof}


\begin{thebibliography}{99}
\bibitem{footnote}For a linear map
$\Lambda:\mathcal{B}(\mathcal{H})\to \mathcal{B}(\mathcal{K})$ we define the
dual map as a linear map $\Lambda^{\dagger}:\mathcal{B}(\mathcal{K})\to
\mathcal{B}(\mathcal{H})$ that satisfies
$\Tr[X\Lambda(Y)]=\Tr[\Lambda^{\dagger}(X)Y]$ for any $X\in
\mathcal{B}(\mathcal{K})$ and $Y\in \mathcal{B}(\mathcal{H})$. Recall that if
$\Lambda$ is positive (in particular completely positive), its
dual $\Lambda^{\dagger}$ is also positive. Moreover, if $\Lambda$
is trace-preserving, i.e., $\Tr[\Lambda(X)]=\Tr X$ for any $X$,
the dual map $\Lambda^{\dagger}$ is unital, i.e.,
$\Lambda^{\dagger}(\mathbbm{1}_{\mathcal{K}})=\mathbbm{1}_{\mathcal{H}}$
with $\mathbbm{1}_{\mathcal{X}}$ being the identity operator
acting on $\mathcal{X}$.
\end{thebibliography}
\end{document}